\documentclass{article}
\usepackage{fullpage}
\usepackage{times}
\usepackage{color}
\usepackage{amsmath,amssymb,amsthm, amsfonts, enumitem,centernot,mathtools,extarrows,mathdots,appendix}
\usepackage{comment,MnSymbol}
\usepackage{physics}
\usepackage{tikz-cd}
\usepackage{algorithm, algorithmic}
\usepackage[hidelinks]{hyperref}
\usepackage[nameinlink,capitalize]{cleveref}

\newcommand{\PRECOMPUTE}{\item[\textbf{Precomputed:}]}

\newtheorem{theorem}{Theorem}
\newtheorem{lemma}{Lemma}
\newtheorem{defn}{Definition}
\newtheorem{proposition}{Proposition}

\title{New Identity for Cayley's First Hyperdeterminant with Applications to Symmetric Tensors and Entanglement}
\author{Isaac Dobes}

\begin{document}

\maketitle

\begin{abstract}
    In this article, a new formula for computing Cayley's first hyperdeterminant in terms of the Levi-Civita symbol is given. It is then shown that this formula can be used to compute the hyperdeterminant of symmetric tensors in polynomial time with respect to their order (assuming fixed side length). Applications to quantifying the entanglement of states of bosonic quantum systems are then discussed. Additionally, in order to obtain the fast calculation of the hyperdeterminant on symmetric tensors, generalized elimination and duplication matrices are defined and their explicit formulas are derived. 
\end{abstract}

\section{Introduction}
Originally introduced in 1844 by Sir Arthur Cayley \cite{cayley1844theory}, Cayley's first hyperdeterminant (also known as the combinatorial hyperdeterminant) was nearly forgotten for the next 150 years. Indeed, only in the last 30 years has there been significant research interest into the applications of Cayley's first hyperdeterminant, and in this time numerous applications to both mathematics and physics have been uncovered! For example, in 1997 P. Zappa expressed Cayley's first hyperdeterminant in terms of the difference in even and odd Latin squares of order $N$ \cite{zappa1997cayley}, in 2003 J. G. Luque and J. Y. Thibon relate Cayley's first hyperdeterminant of Hankel hypermatrices to Selberg Integrals \cite{luque2003hankel}, in 2008 S. Matsumoto expresses a particular type of Jack function in terms of Cayley's first hyperdeterminant \cite{matsumoto2008hyperdeterminantal}, in 2021 P. Lammers utilizes Cayley's first hyperdeterminant to characterize a generalized Kasteleyn theory \cite{lammers2021generalisation}, also in 2021 A. Amanov and D. Yeliussizov used Cayley's first hyperdeterminant to establish lower bounds of certain tensor ranks \cite{amanov2023tensor}, and in 2024/2025 I. Dobes and N. Jing prove that Cayley's first hyperdeterminant generalizes the concurrence, an important entanglement measure in quantum information \cite{dobes2024qubits,dobes2025cayley}. Therefore, while one may initially consider Cayley's first hyperdeterminant to be a naive generalization of the usual determinant, it is in fact an important mathematical object with a wide variety of applications. 

In this paper, we derive a new formula for computing Cayley's first hyperdeterminant in terms of the Levi-Civita symbol. For an order $N$ tensor $\mathcal{A}\in \mathbb{K}^{d\times\dots d}$, \textbf{Cayley's first hyperdeterminant} (also known as the \textbf{combinatorial hyperdeterminant}) on $\mathcal{A}$ is typically defined as 
\[\mathrm{hdet}(\mathcal{A}) := \frac{1}{d!}\sum\limits_{\sigma_1,\dots,\sigma_N\in S_d}\left(\prod\limits_{k=1}^N\mathrm{sgn}(\sigma_k)\right)\prod\limits_{i=1}^da_{\sigma_1(i)\dots\sigma_N(i)}.\]
If $N$ is odd, then $\mathrm{hdet}(\mathcal{A})$ is identically $0$ \cite{lim2013tensors}, however if $N$ is even then $\mathrm{hdet}(\mathcal{A})$ is non-trivial in general, and furthermore we have that
\[\mathrm{hdet}(\mathcal{A}) = \sum\limits_{\sigma_2,\dots,\sigma_N\in S_d}\left(\prod\limits_{k=2}^N\mathrm{sgn}(\sigma_k)\right)\prod\limits_{i=1}^da_{i\sigma_2(i)\dots\sigma_N(i)}.\]
The $d$-dimensional Levi-Civita symbol $\varepsilon = [\varepsilon_{i_1i_2\dots i_d}]\in \mathbb{K}^{d\times\dots\times d}$, which represents an important antisymmetric tensor in geometry and physics (e.g. see \cite{gilmore2008lie}), is the order $d$ tensor such that 
\[\varepsilon_{i_1i_2\dots i_d} = 
    \begin{cases}
        \begin{rcases}
            1 & \text{if }(i_1i_2\dots i_d)\text{ is an even permutation of the tuple }(12\dots d) \\ 
            -1 & \text{if }(i_1i_2\dots i_d)\text{ is an odd permutation of the tuple }(12\dots d) \\
            0 & \text{otherwise}
        \end{rcases}
    \end{cases}\]
In this article, we show that if $\mathcal{A}\in \mathbb{K}^{d\times...\times d}$ is an arbitrary order $N$ tensor, then Cayley's first hyperdeterminant of $\mathcal{A}$ is equal to the multilinear matrix product
\[\left(\frac{1}{d!}\bigotimes\limits_{k=1}^N\varepsilon\right)*(\mathrm{hvec}(\mathcal{A}),\dots,\mathrm{hvec}(\mathcal{A})),\]
where $\varepsilon$ denotes the $d$-dimensional Levi-Civita symbol, "$\bigotimes$" denotes the tensor Kronecker product, "$*$" denotes multilinear matrix multiplication, and $\mathrm{hvec}(\mathcal{A})$ denotes a vectorization of the tensor $\mathcal{A}$. The major benefit to this new identity for computing Cayley's first hyperdeterminant applies to symmetric tensors. 

In general, the run time cost for computing Cayley's first hyperdeterminant (which we denote as $\mathrm{hdet}$) of an arbitrary cubical tensor of order $N$ with side length $d$ is VNP-hard \cite{hillar2013most}. Indeed, with current state-of-the-art methods the run time cost for computing $\mathrm{hdet}$ is exponential in both $N$ and $d$; even after fixing one of the variables, it is still exponential with respect to the other \cite{amanov2024enhanced}. However, after generalizing the notions of half-vectorization and duplication matrices to apply to symmetric hypermatrices of order $N$ with side length $d$, it is then shown in this article that the run time cost for computing $\mathrm{hdet}$ on symmetric hypermatrices is in fact polynomial with respect to $N$, assuming $d$ fixed. We then discuss how this result can be used to efficiently calculate the $2n$-way entanglement of states of bosonic quantum systems. Lastly in the appendix, we derive an explicit formula for the generalized duplication matrix, which the fast calculation of $\mathrm{hdet}$ on symmetric hypermatrices depends on. 

\section{Hypermatrix Algebra}
\subsection{Preliminaries}
In this paper we will use $\mathbb{K}$ to denote an arbitrary field of characteristic $0$, and unless otherwise specified we will always assume that the base field is $\mathbb{K}$. The material in this and the next section holds for any field $F$ of any characteristic, however from Section 4 onward we will need to assume the characteristic is $0$, so for notational consistency and convenience we just go ahead and write everything in terms of $\mathbb{K}$. Also, we utilize the convention in combinatorics that for a positive integer $n$, $[n]$ denotes the set $\{1,2,\dots,n\}$. 

A \textbf{hypermatrix} of order $N$ with dimensions $n_1\times n_2\times\dots\times n_N$ is an $N$-dimensional array; specifically, an element in $\mathbb{K}^{n_1\times n_2\times\dots\times n_N}$. Since every higher-dimensional array may be viewed as the coordinate array of some tensor in a fixed basis, we identify these arrays with their corresponding tensors and adopt the more common convention of referring to the arrays themselves as tensors. A tensor $A \in \mathbb{K}^{n_1\times n_2\times\dots\times n_N}$ is called \textbf{cubical} if $n_1=n_2=\dots=n_N =: d$, and in such case we call the number $d$ the \textbf{side length} of the tensor. An order $N$ cubical tensor $\mathcal{A} = [a_{i_1\dots i_N}]\in \mathbb{K}^{d\times\dots\times d}$ is called \textbf{symmetric} if $a_{i_1\dots i_N} = a_{i_{\sigma(1)}\dots i_{\sigma(N)}}$ for every $\sigma\in S_N$.
\begin{defn}[Segre Outer Product]
    Suppose $\mathcal{A} = [a_{i_1\dots i_N}]\in \mathbb{K}^{n_1\times...\times n_N}$ and $\mathcal{B} = [b_{j_1\dots j_M}]\in \mathbb{K}^{m_1\times\dots\times m_M}$ are tensors of order $N$ and $M$ respectively. The \textbf{Segre outer product} of $\mathcal{A}$ with $\mathcal{B}$, denoted $\mathcal{A}\circ \mathcal{B}$, is the order $N+M$ tensor whose $(i_1,\dots,i_N,j_1,\dots,j_M)$-coordinate is given by $a_{i_1\dots i_N}b_{j_1\dots j_M}$.
\end{defn}

\begin{defn}[Multilinear Matrix Multiplication]\cite{lim2013tensors}
    Suppose $\mathcal{A} = [a_{i_1\dotsi_N}] \in \mathbb{K}^{n_1\times\dots\times n_N}$, and $X^{(k)} = [x^{(k)}_{ij}]\in \mathbb{K}^{n_k\times m_k}$ for $k=1,...,N$. The \textbf{(right) multilinear matrix multiplication} of $\mathcal{A}$ with the tuple $(X^{(1)},\dots,X^{(N)})$ is given by the tensor $A*(X^{(1)},\dots,X^{(N)}) =: \mathcal{A}'$ in $\mathbb{K}^{m_1\times\dots\times m_N}$ whose $(j_1,\dots,j_N)$-coordinate is given by 
    \[a_{j_1\dots j_N}' \ = \sum\limits_{k_1,\dots,k_N=1}^{n_1,\dots,n_N}a_{k_1\dots k_N}x_{k_1j_1}^{(1)}\dots x_{k_Nj_N}^{(N)}\]
\end{defn}
Note also that it is common in the literature for multilinear matrix multiplication to be defined as a left action, however we define it as a right action since this convention leads to slightly cleaner formulas (atleast in this article) and makes our subsequent generalizations of known objects (namely, elimination and duplication matrices) more faithful to their classical counterparts. 
\begin{proposition}[\cite{lim2013tensors}]
    Let $\alpha,\beta\in \mathbb{K}$; $X^{(1)},Y^{(1)}\in \mathbb{K}^{m_1\times n_1}$,\dots m $X^{(N)},Y^{(N)}\in \mathbb{K}^{m_N\times n_N}$; and $\mathcal{A},\mathcal{B}\in \mathbb{K}^{n_1\times\dots\times n_N}$. Multilinear matrix multiplication satisfies the following properties. 
    \begin{enumerate}
        \item Multilinearity:  
        \[(\alpha \mathcal{A} + \beta \mathcal{B})*(X^{(1)},\dots,X^{(N)}) = \alpha \mathcal{A}*(X^{(1)},\dots,X^{(N)}) + \beta \mathcal{B}*(X^{(1)},\dots,X^{(N)})\]
        and 
        \[\mathcal{A}*[\alpha(X^{(1)},\dots,X^{(N)})+\beta(Y^{(1)},\dots,Y^{(N)})] = \alpha \mathcal{A}*(X^{(1)},\dots,X^{(N)}) + \beta \mathcal{A}*(Y^{(1)},\dots,Y^{(N)}).\]
        \item Multiplicative: if $Z^{(1)}\in \mathbb{K}^{n_1\times p_1}$,$\dots$, $Z^{(N)}\in \mathbb{K}^{n_N\times p_N}$, then 
        \[\mathcal{A}*(X^{(1)}Z^{(1)},\dots,X^{(N)}Z^{(N)}) = [\mathcal{A}*(X^{(1)},\dots,X^{(N)})]*(Z^{(1)},\dots,Z^{(N)}).\]
    \end{enumerate}
\end{proposition}
\begin{defn}[Tensor Kronecker Product]
    The \textbf{tensor Kronecker product} of two order $N$ hypermatrices $\mathcal{A} = [a_{i_1\dotsi_N}]\in \mathbb{K}^{m_1\times\dots\times m_N}$ and $\mathcal{B} = [b_{i_1\dotsi_N}]\in \mathbb{K}^{n_1\times\dots\times n_N}$ is defined as the order $N$ block tensor $\mathcal{A}\otimes \mathcal{B}\in \mathbb{K}^{m_1n_1\times\dots\times m_Nn_N}$ whose $(i_1,\dots,i_N)$-block is the tensor $a_{i_1\dots i_N}\mathcal{B}$. 
\end{defn}
The tensor Kronecker product is compatible with multilinear matrix multiplication in the following way. 
\begin{proposition}\cite{pickard2024kronecker}\label{TKP Property}
    Let $\mathcal{A} \in \mathbb{K}^{m_1\times\dots\times m_N}$ and $\mathcal{B}\in \mathbb{K}^{n_1\times\dots\times n_N}$ be two hypermatrices of order $N$, and suppose $X^{(i)}\in \mathbb{K}^{m_i\times p_i}$ and $Y^{(i)}\in \mathbb{K}^{n_i\times q_i}$ for each $i\in [N]$. Then 
    \[(\mathcal{A}*(X^{(1)},\dots,X^{(N)}))\otimes (\mathcal{B}*(Y^{(1)},\dots,Y^{(N)})) = (\mathcal{A}\otimes \mathcal{B})*(X^{(1)}\otimes Y^{(1)},\dots,X^{(N)}\otimes Y^{(N)}).\]
\end{proposition}

\subsection{Vectorizing Tensors}
Note that if $e_{i_k}\in \mathbb{K}^{n_k}$ for $k=1,\dots N$ are standard basis vectors, then $e_{i_1}\circ\dots\circ e_{i_N}\in \mathbb{K}^{n_1\times\dots\times n_N}$ is the order $N$ tensor with a $1$ in its $(i_1,\dots,i_N)$-entry and $0$'s elsewhere, hence very tensor $\mathcal{A}\in \mathbb{K}^{n_1\times...\times n_N}$ may be uniquely written as 
\[\mathcal{A} = \sum\limits_{i_1,\dots i_N=1}^{n_1,\dots n_N}a_{i_1\dots i_N}e_{i_1}\circ\dots e_{i_N}.\]
Moreover, $e_{i_1}\otimes...\otimes e_{i_{i_N}}\in \mathbb{K}^{n_1\cdot\dots\cdot n_N}$ is the $\left(i_N+n_N(i_{N-1}-1)+\dots +n_Nn_{N-1}\cdot\dots\cdot n_1(i_1-1)\right)^{th}$ element in the standard ordered basis of $\mathbb{K}^{n_1\cdot\dots\cdot n_N}$. Therefore, extending linearly the map on basis elements
\begin{align*}
    \mathrm{hvec}:\mathbb{K}^{n_1\times\dots\times n_N}&\longrightarrow \mathbb{K}^{n_1\cdot\dots\cdot n_N} \\
    e_{i_1}\circ...\circ e_{i_N}&\mapsto e_{i_1}\otimes...\otimes e_{i_N}
\end{align*}
is a linear isomorphism; we denote this map as $\mathrm{hvec}$ (for \textbf{hypermatrix vectorization}).\footnote{Note that in the special case of $N=2$, in which case $\mathcal{A}$ may be thought of as a matrix, $\mathrm{hvec}(\mathcal{A}) = \mathrm{vec}(\mathcal{A}^T)$, where $\mathrm{vec}$ is the typical vectorization map defined on matrices.} 

Next, note that if $\mathcal{A}\in \mathbb{K}^{d\times \dots \times d}$ is symmetric, there is a natural equivalence relation on the coordinates of $\mathcal{A}$: $a_{i_1\dots i_N}\sim a_{j_1\dots j_N} \iff j_k = i_{\sigma(k)}$ for each $k\in [N]$, for some $\sigma\in S_N$. Hence, we may recover all of the coordinates of $\mathcal{A}$ by considering only the equivalence class representatives $a_{i_1\dots i_N}$ such that $1\leq i_1\leq\dots\leq i_N\leq d$. This motivates the following definition.
\begin{defn}
    Let $\mathcal{A}\in \mathbb{K}^{d\times\dots d}$ be a symmetric order $N$ tensor. The \textbf{$\boldsymbol{1/N}$-hypermatrix vectorization} of $\mathcal{A}$, denoted as $\mathrm{hvec}_{1/N}(\mathcal{A})$, is defined to be the vector consisting of the equivalence class representatives mentioned above, ordered in the lexicographical order from least to greatest. In other words, 
    \[\mathrm{hvec}_{1/N}(\mathcal{A}) := \left[\begin{array}{c}
        a_{11\dots 11} \\
        a_{11\dots 12} \\
        \vdots \\
        a_{11\dots 1d} \\
        a_{11\dots 22} \\
        \vdots \\
        a_{11\dots 2d} \\
        a_{11\dots 33} \\
        \vdots \\
        \vdots \\
        a_{d-1,d\dots dd} \\
        a_{dd\dots dd}
    \end{array}\right]\]
\end{defn}
Since $\dim(\mathrm{Sym}^d(\mathbb{K}^d)) = \binom{d+N-1}{N}$, it follows that $\mathrm{hvec}_{1/N}(\mathcal{A})\in \mathbb{K}^{\binom{d+N-1}{N}}$. 

\section{Main Result}
\subsection{Cayley's First Hyperdeterminant in Terms of the Levi-Civita Symbol}
First we note the following observation. 
\begin{lemma}\label{Hyperdeterminant in Terms of LC}
    Let $\mathcal{A} = [a_{i_1\dots i_N}]\in \mathbb{K}^{d\times\dots \times d}$ be a cubical tensor of order $N$ with side length $d$. Then 
    \begin{equation}\label{Hyperdeterminant as Big Sum w/LC}
        \mathrm{hdet}(\mathcal{A}) = \frac{1}{d!}\sum\limits_{\substack{i_k^{(1)},\dots,i_k^{(d)}=1 \\ k\in [N]}}^d\prod\limits_{k=1}^N\varepsilon_{i_k^{(1)},\dots,i_k^{(d)}}\prod\limits_{l=1}^da_{i_1^{(l)},\dots,i_N^{(l)}}.
    \end{equation}
\end{lemma}
\begin{proof}
    The product $\prod\limits_{k=1}^N\varepsilon_{i_1^{(k)},\dots,i_d^{(k)}}$ is nonzero only if for each $k\in [N]$, $(i_1^{(k)},\dots,i_d^{(k)})$ is a permutation of $S_d$. Therefore, assuming $(i_1^{(k)},\dots,i_d^{(k)})$ is in fact a permutation of $S_d$, and denoting said permutation as $\sigma_k$, it then follows that the sum in equation \eqref{Hyperdeterminant as Big Sum w/LC} can be rewritten as 
    \[\frac{1}{d!}\sum\limits_{\substack{\sigma_k\in S_d \\ k\in [N]}}\left(\prod\limits_{k=1}^N\varepsilon_{\sigma_k(1)\dots\sigma_k(d)}\right)\prod\limits_{l=1}^da_{\sigma_1(l)\dots\sigma_N(l)}.\]
    By definition $\varepsilon_{\sigma_k(1)\dots\sigma_k(d)} = \mathrm{sgn}(\sigma_k)$, hence the sum in equation \eqref{Hyperdeterminant as Big Sum w/LC} is in fact equal to
    \[\frac{1}{d!}\sum\limits_{\substack{\sigma_k\in S_d \\ k\in [N]}}\left(\prod\limits_{k=1}^N\mathrm{sgn}(\sigma_k)\right)\prod\limits_{l=1}^da_{\sigma_1(l)\dots\sigma_N(l)},\]
    which is precisely the formula for $\mathrm{hdet}(\mathcal{A})$.  
\end{proof}
With the above lemma, we may now prove the main result of this article. 
\begin{theorem}\label{Main Result}
    Let $\mathcal{A} = [a_{i_1\dots i_N}]\in \mathbb{K}^{d\times\dots\times d}$ be a cubical tensor of order $N$ with side length $d$. Then 
    \[\mathrm{hdet}(\mathcal{A}) = \left(\frac{1}{d!}\bigotimes\limits_{k=1}^N\varepsilon\right)*\big(\mathrm{hvec}(\mathcal{A}),\dots,\mathrm{hvec}(\mathcal{A})\big),\]
    where $\varepsilon$ denotes the $d$-dimensional Levi-Civita symbol, $\bigotimes$ denotes the tensor Kronecker product, and the tuple \\
    $\big(\mathrm{hvec}(\mathcal{A}),\dots,\mathrm{hvec}(\mathcal{A})\big)$ has length $d$. 
\end{theorem}
\begin{proof}
    Observe
    \begin{align*}
        &\left(\frac{1}{d!}\bigotimes\limits_{k=1}^N\varepsilon\right)*\big(\mathrm{hvec}(\mathcal{A}),\dots,\mathrm{hvec}(\mathcal{A})\big) \\
        &= \left(\frac{1}{d!}\bigotimes\limits_{k=1}^N\varepsilon\right)*\left(\sum\limits_{i_1^{(1)},\dots,i_N^{(1)}=1}^da_{i_1^{(1)}\dots i_N^{(1)}}e_{i_1^{(1)}}\otimes\dots\otimes e_{i_N^{(1)}},\dots,\sum\limits_{i_1^{(d)},\dots,i_N^{(d)}=1}^da_{i_1^{(d)}\dots i_N^{(d)}}e_{i_1^{(d)}}\otimes\dots\otimes e_{i_N^{(d)}}\right),\quad \text{by definition} \\
        &= \frac{1}{d!}\sum\limits_{\substack{i_1^{(l)},\dots,i_N^{(l)}=1 \\ l\in [d]}}^d\prod\limits_{l=1}^da_{i_1^{(l)}\dots i_N^{(l)}}\left(\bigotimes_{k=1}^N\varepsilon*(e_{i_1^{(1)}}\otimes\dots\otimes e_{i_N^{(1)}},\dots,e_{i_1^{(d)}}\otimes\dots\otimes e_{i_N^{(d)}})\right),\quad \text{by multilinearity} \\
        &= \frac{1}{d!}\sum\limits_{\substack{i_1^{(l)},\dots,i_N^{(l)}=1 \\ l\in [d]}}^d\prod\limits_{l=1}^da_{i_1^{(l)}\dots i_N^{(l)}}\big(\varepsilon*(e_{i_1^{(1)}},\dots,e_{i_1^{(d)}})\big)\otimes\dots\otimes \big(\varepsilon*(e_{i_N^{(1)}},\dots,e_{i_N^{(d)}})\big),\quad \text{by \cref{TKP Property}} \\
        &= \frac{1}{d!}\sum\limits_{\substack{i_1^{(l)},\dots,i_N^{(l)}=1 \\ l\in [d]}}^d\prod\limits_{l=1}^da_{i_1^{(l)}\dots i_N^{(l)}}\varepsilon_{i_1^{(1)}\dots i_1^{(d)}}\dots\varepsilon_{i_N^{(1)}\dots i_N^{(d)}} \\
        &= \frac{1}{d!}\sum\limits_{\substack{i_k^{(1)},\dots,i_k^{(d)}=1 \\ k\in [N]}}^d\prod\limits_{k=1}^N\varepsilon_{i_k^{(1)}\dots i_k^{(d)}}\prod\limits_{l=1}^da_{i_1^{(l)}\dots i_N^{(l)}} \\
        &= \mathrm{hdet}(\mathcal{A}),\quad \text{by \cref{Hyperdeterminant in Terms of LC}.}
    \end{align*}
\end{proof}
Note in particular that when $N$ is odd, the anti-symmetry of $\varepsilon$ implies that the product in \cref{Main Result} is identically $0$, in agreement with the fact that Cayley's first hyperdeterminant is identically $0$ on all odd order hypermatrices. 

\subsection{Run Time Cost of Cayley's First Hyperdeterminant via Levi-Civita Symbol}
Let $\varepsilon$ denote the $d$-dimensional Levi-Civita symbol, which recall is a cubical tensor of order $d$ with side length $d$. It then follows that the run time cost for computing the tensor Kronecker product $\bigotimes\limits_{k=1}^N\varepsilon$ is $\mathcal{O}(d^{Nd})$. Suppose also that $\mathcal{A}$ is a cubical tensor of order $N$ with side length $d$. Then $\mathrm{hvec}(\mathcal{A})\in \mathbb{K}^{d^N}$ and so the run time cost for preparing $\mathrm{hvec}(\mathcal{A})$ is $\mathcal{O}\big(d^N\big)$. By \cref{Main Result}, 
\begin{equation}\label{LC hyperdeterminant computation}
    \mathrm{hdet}(\mathcal{A}) = \left(\frac{1}{d!}\bigotimes_{k=1}^N \varepsilon\right) * \big(\mathrm{hvec}(\mathcal{A}),\dots, \mathrm{hvec}(\mathcal{A})\big),
\end{equation}
which yields a run time cost of $\mathcal{O}(d^{2Nd})$ due to the evaluation of the multilinear matrix product of $\bigotimes\limits_{k=1}^N\varepsilon\in \mathbb{K}^{d^N\times\dots\times d^N}$ with $d$ copies of $\mathrm{hvec}(\mathcal{A})\in \mathbb{K}^{d^N}$ (note that asymptotically the initial cost of preparing $\mathrm{hvec}(\mathcal{A})$ is marginal in comparison so this may be omitted in the cost calculation). This approach is slower than the current state-of-the-art algorithms for computing Cayley’s first hyperdeterminant, which is currently $\mathcal{O}\big(2^{d(N-1)}d^{N-1}\big)$ according to \cite{amanov2024enhanced}. 

Nonetheless, this method has one key advantage: the tensor Kronecker product $\bigotimes\limits_{k=1}^N \varepsilon$ depends only on $d$ and $N$, not on the specific tensor $\mathcal{A}$. Consequently, it may be computed once and reused for all other subsequent hyperdeterminant calculations. Assuming that $\bigotimes\limits_{k=1}^N\varepsilon$ is precomputed and stored, the run time cost of computing $\mathrm{hdet}(\mathcal{A})$ via identity \eqref{LC hyperdeterminant computation} from \cref{Main Result} reduces to $\mathcal{O}(d^{Nd})$. 

We now compare this with the time complexity of the current state of the art method. Consider the ratio of the two run times: 
\[\frac{d^{Nd}}{2^{d(N-1)}d^{N-1}} = d^{Nd-(N-1)}2^{-d(N-1)}\]
Taking logarithms, we obtain
\begin{equation}\label{log comparison}
\ln\left(d^{Nd - (N - 1)} \cdot 2^{-d(N - 1)}\right) = \Big(\big(Nd - (N-1)\big)\log_2(d) - (Nd - d)\Big)\ln(2).
\end{equation}
In general, \eqref{log comparison} will be positive except for when $d=2$ and $N\geq 3$. Consequently, computing $\mathrm{hdet}(\mathcal{A})$ via identity \eqref{LC hyperdeterminant computation} is in general slower than current state of the art methods, except for in the special case where $d=2$ and $N\geq 3$. However, when it comes to \textit{symmetric} hypermatrices, the benefit to computing Cayley's first hyperdeterminant via identity \eqref{LC hyperdeterminant computation} yields a far more substantial advantage, which we explain in the next section. 

\section{Fast Calculation for $\mathrm{hdet}$ on Symmetric Hypermatrices}
\subsection{Generalized Elimination and Duplication Matrices}
In \cite{magnus1980elimination}, Magnus and Neudecker introduce the notion of elimination and duplication matrices, which they define as follows: if $A$ is any $d\times d$ symmetric matrix, then the \textbf{elimination matrix}, denoted $L_d$, is the unique $\frac{d(d+1)}{2}\times d^2$ matrix such that 
\[L_d\mathrm{vech}(A) = \mathrm{vec}(A);\]
on the other hand, the duplication matrix, denoted $D_d$, is the unique $d^2\times \frac{d(d+1)}{2}$ matrix such that 
\[D_d\mathrm{vech}(A) = \mathrm{vec}(A).\]
It was shown in \cite{magnus1980elimination} that for each positive integer $d\geq 2$, 
\[L_d = \sum\limits_{d\geq i\geq j\geq 1}u_{ij}\mathrm{vec}(E_{ij})^T\]
and 
\[D_d = \sum\limits_{d\geq i\geq j\geq 1}\mathrm{vec}(T_{ij})u_{ij}^T,\]
where $u_{ij}$ is the $\frac{1}{2}d(d+1)$-dimensional unit vector with a $1$ in its $\big((j-1)n+i-\frac{1}{2}j(j-1)\big)^{th}$-coordinate and $0$'s everywhere else, and $T_{ij}$ is the $d\times d$ matrix with a $1$ in its $(i,j)$ and $(j,i)$ coordinates, and $0$'s everywhere else (i.e. $T_{ij} = E_{ij}+E_{ji}$ when $i\neq j$, otherwise $T_{ii} = E_{ii}$).

We may extending the notion of elimination and duplication matrices to symmetric hypermatrices, which we define as follows:
\begin{defn}
    The \textbf{generalized elimination matrix} $L_d^{(N)}$ is the unique $\binom{d+N-1}{N}\times d^N$ matrix such that 
    \[L_d^{(N)}\mathrm{hvec}(\mathcal{A}) = \mathrm{hvec}_{1/N}(\mathcal{A}),\]
    and the \textbf{generalized duplication matrix} $D_d^{(N)}$ is the unique $d^N\times \binom{d+N-1}{N}$ matrix such that 
    \[D_d^{(N)}\mathrm{hvec}_{1/N}(\mathcal{A})=\mathrm{hvec}(\mathcal{A}),\]
    where $\mathcal{A}$ is any arbitrary symmetric tensor of order $N$ and side length $d$. 
\end{defn}
\begin{proposition}
    Let $N,d\geq 2$ be positive integers. The generalized elimination matrix is given by 
    \[L_d^{(N)} = \sum\limits_{1\leq i_1\leq\dots\leq i_N\leq d}u_{i_1\dots i_N}\mathrm{hvec}(E_{i_1\dots i_N})^T\]
    and the generalized duplication matrix $D_d^{(N)}$ is given by 
    \[D_d^{(N)} = \sum\limits_{1\leq i_1\leq\dots \leq i_N\leq d}\mathrm{hvec}(T_{i_1\dots i_N})u_{i_1\dots i_N}^T,\]
    where $u_{i_1\dots i_N}$ is the unit vector of length $\binom{d+N-1}{N}$ with a $1$ in its $\left(\binom{d+N-1}{N}-\sum\limits_{k=1}^N\binom{d+N-k-i_k}{N-k+1}\right)^{th}$-coordinate and $0$'s elsewhere, $E_{i_1\dots i_N} := e_{i_1}\circ\dots \circ e_{i_N}$, and $T_{i_1\dots i_N} := \sum\limits_{\sigma\in \mathrm{Orb}(i_1,\dots ,i_N)}E_{i_{\sigma(1)}\dots i_{\sigma(N)}}$. 
\end{proposition}
\begin{proof}
    It is well known the number of tuples $(i_1,\dots ,i_N)$ such that $1\leq i_1\leq\dots\leq i_N\leq d$ is equal to $\binom{d+N-1}{N}$, hence it follows that the number of tuples $(i_1,\dots ,i_N)$ such that $k\leq i_1\leq\dots\leq i_N\leq d$ is equal to $\binom{d+N-k}{N}$. Listing all tuples $(i_1,...,i_N)$ such that $1\leq i_1\leq\dots\leq i_N\leq d$ in lexicographical order, $(j_1,\dots ,j_N) > (i_1,\dots, i_N)$ if and only if there exists $k\in [N]$ such that $j_k > i_k$ and $j_l = i_l$ for all $l$ strictly less than $k$; in particular, the last $N-k+1$ entries of $(j_1,...,j_N)$ are subject only to the restriction that $i_k+1 \leq j_k\leq j_{k+1}\leq\dots\leq j_N\leq d$, from which it follows that there are exactly $\sum\limits_{k=1}^N\binom{d+(N-k+1)-(i_k+1)}{N-k+1} = \sum\limits_{k=1}^N\binom{d+N-k-i_k}{N-k+1}$ many tuples greater than $(i_1,\dots ,i_N)$. Thus, the placement of the tuple $(i_1,\dots ,i_N)$, when listed in lexicographical order from least to greatest, is given by $\binom{d+N-1}{N}-\sum\limits_{k=1}^N\binom{d+N-k-i_k}{N-k+1}$. 
    
    Denote $u_{i_1\dots i_N}\in \mathbb{K}^{\binom{d+N-1}{N}}$ to be the unit vector with a $1$ in the placement of $(i_1,\dots,i_N)$ and $0$'s elsewhere. Then for any order $N$ symmetric tensor $\mathcal{A} = [a_{i_1\dots i_N}]\in \mathbb{K}^{d\times\dots \times d}$, 
    \[\mathrm{hvec}_{1/N}(\mathcal{A}) = \sum\limits_{1\leq i_1\leq\dots\leq i_N\leq d}a_{i_1\dots i_N}u_{i_1\dots i_N}.\] Setting $E_{i_1...i_N} := e_{i_1}\circ\dots \circ e_{i_N}$ with each $e_{i_k}$ a standard ordered basis element in $\mathbb{K}^d$, it follows by definition and linearity of Kronecker products that 
    \[\sum\limits_{1\leq i_1\leq\dots\leq i_N\leq d}u_{i_1\dots i_N}\mathrm{hvec}(E_{i_1\dots i_N})^T\mathrm{hvec}(\mathcal{A}) = \mathrm{hvec}_{1/N}(\mathcal{A}).\]
    Similarly, if $T_{i_1\dots i_N} := \sum\limits_{\sigma\in \mathrm{Orb}(i_1,...,i_N)}E_{i_{\sigma(1)}\dots i_{\sigma(N)}}$, summing over the orbit of each equivalence class representative yields 
    \[\sum\limits_{1\leq i_1\leq\dots i_N\leq d}T_{i_1\dots i_N} = \sum\limits_{i_1,\dots i_N=1}^dE_{i_1\dots i_N},\]
    hence by orthogonality and the equation above we have that 
    \begin{align*}
        \sum\limits_{1\leq i_1\leq\dots \leq i_N\leq d}\mathrm{hvec}(T_{i_1\dots i_N})u_{i_1\dots i_N}^T\mathrm{hvec}_{1/N}(\mathcal{A}) &= \sum\limits_{1\leq i_1\leq \dots \leq i_N\leq d}\mathrm{hvec}(T_{i_1...i_N})a_{i_1\dots i_N} \\
        &= \sum\limits_{i_1,\dots i_N=1}^da_{i_1...i_N}\mathrm{hvec}(E_{i_1\dots i_N}) \\
        &= \mathrm{hvec}(\mathcal{A}),
    \end{align*}
    completing the proof. 
\end{proof}

\subsection{Time Complexity Reduction via $D_d^{(N)}$}
Now, let $\mathcal{A}$ be any symmetric order $N$ tensor and side length $d$. Then by \cref{Main Result}, Cayley's first hyperdeterminant of $\mathcal{A}$ is given by 
\begin{align*}
    \left(\frac{1}{d!}\bigotimes\limits_{k=1}^N\varepsilon\right)*(\mathrm{hvec}(\mathcal{A}),\dots,\mathrm{hvec}(\mathcal{A})) &= \left(\frac{1}{d!}\bigotimes\limits_{k=1}^N\varepsilon\right)*\big(D_d^{(N)}\mathrm{hvec}_{1/N}(\mathcal{A}),\dots,D_d^{(N)}\mathrm{hvec}_{1/N}(\mathcal{A})\big) \\
    &= \left(\frac{1}{d!}\bigotimes\limits_{k=1}^N\varepsilon*(D_d^{(N)},\dots,D_d^{(N)})\right)*(\mathrm{hvec}_{1/N}(\mathcal{A}),\dots,\mathrm{hvec}_{1/N}(\mathcal{A})). 
\end{align*}
Denote the product $\frac{1}{d!}\bigotimes\limits_{k=1}^N\varepsilon*\left(D_d^{(N)},\dots,D_d^{(N)}\right)$ as $\mathcal{E}_d^{(N)}$. Note that $\mathcal{E}_d^{(N)}$ is a cubical tensor with order $d$ and side length $\binom{d+N-1}{N}$. Using this notation, it follows that 
\begin{equation}\label{Hyperdeterminant Computation - Symmetric Case}
    \mathrm{hdet}(A) = \mathcal{E}_d^{(N)}*(\mathrm{hvec}_{1/N}(A),\dots,\mathrm{hvec}_{1/N}(A)).
\end{equation}
The matrix $D_d^{(N)}$, and hence also the tensor $\mathcal{E}_d^{(N)}$, depends only on $d$ and $N$. Therefore similarly as before, $\mathcal{E}_d^{(N)}$ may be computed once and reused for all subsequent hyperdeterminant calculations. Assuming $\mathcal{E}_d^{(N)}$ is precomputed and stored, the run time cost of $\mathrm{hdet}$ on symmetrical hypermatrices is $\mathcal{O}\left(\binom{d+N-1}{N}^d\right)$ due to the multilinear matrix product of $\mathcal{E}_d^{(N)}$ with the $d$ copies of $\mathrm{hvec}_{1/N}(\mathcal{A})$ (noting again that asymptotically the initial cost of preparing $\mathrm{hvec}_{1/N}(\mathcal{A})$ is marginal and hence omitted in the cost calculation). Formally, we have the algorithm for computing $\mathrm{hdet}(\mathcal{A})$ for any order $N$ symmetric tensor $\mathcal{A}$ with side length $d$:
\begin{algorithm}[H]
\caption{Fast Calculation of $\mathrm{hdet}$ on Symmetric Hypermatrices}
\begin{algorithmic}[1]\label{Algorithm 1}
\PRECOMPUTE $\mathcal{E}_d^{(N)}\in \mathbb{K}^{\binom{d+N-1}{N}\times\dots\times \binom{d+N-1}{N}}$
\REQUIRE Order $N$ symmetric tensor $\mathcal{A}\in \mathbb{K}^{d\times\dots\times d}$
\ENSURE $\mathrm{hdet}(\mathcal{A})\in \mathbb{K}$
\STATE $\mathbf{a} \leftarrow \mathrm{hvec}_{1/N}(\mathcal{A})$
\STATE $\mathrm{hdet}(\mathcal{A}) \leftarrow \mathcal{E}_d^{(N)}*(\mathbf{a},\dots \mathbf{a})$
\RETURN $\mathrm{hdet}(\mathcal{A})$
\end{algorithmic}
\end{algorithm}
Now, recall that Stirling's approximation says that $n! \sim \sqrt{2\pi n}\left(\frac{n}{e}\right)^n$ for any positive integer $n$ \cite{mermin1984stirling}. Therefore, 
\[\binom{d+N-1}{N} = \frac{(d+N-1)!}{N!(d-1)!}\sim \sqrt{\frac{d+N-1}{2\pi N(d-1)}}\cdot\frac{\left(\frac{d+N-1}{e}\right)^{d+N-1}}{\left(\frac{N}{e}\right)^N\left(\frac{d-1}{e}\right)^{d-1}}.\]
For fixed $d$, we have that $d+N-1\sim N$ and $\sqrt{\frac{d+N-1}{2\pi N(d-1)}}\sim \frac{1}{\sqrt{2\pi(d-1)}}$ as $N\rightarrow \infty$, hence it follows that 
\[\binom{d+N-1}{N}\sim \frac{1}{\sqrt{2\pi(d-1)}}\cdot \left(\frac{e}{N}\right)^N\left(\frac{e}{d-1}\right)^{d-1}\left(\frac{N}{e}\right)^{d+N-1} \propto N^{d-1}.\]
Thus, the complexity of \cref{Algorithm 1} is $\mathcal{O}\left((N^{d-1})^d\right) = \mathcal{O}(N^{d(d-1)})$, which is polynomial in $N$. This is a significant improvement from the current state-of-the art method yielding time complexity $\mathcal{O}(2^{d(N-1)}d^{N-1})$, which is still exponential in $N$ with fixed $d$. For emphasis, we summarize by way of the following theorem. 
\begin{theorem}
    Assuming $\mathcal{E}_d^{(N)}$ is precomputed and stored, the time complexity of \cref{Algorithm 1} $\mathcal{O}(N^{d(d-1)})$, hence polynomial in $N$ for fixed $d$. 
\end{theorem}
The table below provides a clear comparison/summary of the different methods for computing $\mathrm{hdet}$:
\begin{table}[H]
\centering
\begin{tabular}{ccc}
\hline
Method & General Case & Symmetric Case \\
\hline
Prior State-of-the-art \cite{amanov2024enhanced} & $\mathcal{O}(2^{d(N-1)}d^{N-1})$ & $\mathcal{O}(2^{d(N-1)}d^{N-1})$ \\
Naive Application of \cref{Main Result} & $\mathcal{O}(d^{dN})$ & $\mathcal{O}(d^{dN})$ \\
\cref{Algorithm 1} ($\mathcal{E}_d^{(N)}$ precomputed) & \text{N/A} & $\mathcal{O}(N^{d(d-1)})$ \\
\hline
\end{tabular}
\caption{Time complexity comparison for computing Cayley's first hyperdeterminant}
\label{tab:complexity}
\end{table}
The space complexity also needs to be accounted for since we are assuming $\mathcal{E}_d^{(N)}$ is precomputed and stored. Since $\mathcal{E}_d^{(N)}$ is a cubical tensor of order $d$ with side length $\binom{d+N-1}{N}$, its memory cost also grows $\mathcal{O}\left(N^{d(d-1)}\right)$, assuming $d$ is fixed. Note also that $\bigotimes\limits_{k=1}^N\varepsilon$ is a very sparse tensor, and so $\mathcal{E}_d^{(N)}$ is also likely to be quite sparse, in which the case the memory costs may potentially be reduced (possibly significantly). For the sake of brevity we will not consider such potential reductions in space complexity, but we encourage interested readers to investigate this. 

\subsection{Application to Bosonic Quantum Systems}
One of the most elusive phenomena in quantum physics, and at the same time one of the most powerful resources in quantum information, is quantum entanglement. Entangled quantum states arise in composite quantum systems, whose state spaces are tensor products of Hilbert spaces. One of the problems considered in quantum physics and quantum information theory is the quantification of a state's entanglement. In general, this is a complicated endeavor because quantum states in a tensor product of 3 or more Hilbert spaces can be entangled in different, physically inequivalent ways \cite{dur2000three}. Nonetheless, it is widely agreed throughout the literature (a few authoritative sources affirming this include \cite{bennett1996mixed,vedral1997quantifying,vidal2000entanglement,eltschka2014quantifying,bengtsson2017geometry}) that any proposed entanglement measure $E$ should at a bare minimum satisfy the following criteria: 
\begin{enumerate}
    \item $E$ vanishes on separable states,
    \item $E$ is invariant under local unitary transformations, and 
    \item $E$ is non-increasing on average under local operations and classical communications (LOCC). 
\end{enumerate}
The three properties listed above are sometimes referred to as the entanglement axioms. 

Qubits are the simplest type of quantum systems and the fundamental building blocks of quantum computation. Qudits are higher-dimensional analogues of qubits. In particular, for any positive integers $n,d\geq 2$, the state of an $n$-qudit ($n$-qubit when $d=2$) is a unit vector $\ket{\psi}\in(\mathbb{C}^d)^{\otimes n}$, which can be written as a sum
\[\ket{\psi} = \sum\limits_{i_1,\dots,i_n=0}^{d-1}\psi_{i_1\dots i_n}\ket{i_1\dots i_n}\]
such that 
\[\sum\limits_{i_1,\dots,i_n=0}^{d-1}|\psi_{i_1\dots i_n}|^2=1,\]
where $\ket{i_1\dots i_n}$ denotes the Kronecker product $\ket{i_1}\otimes\dots \otimes \ket{i_n}$ and $\ket{i_k}$ denotes the $(i_k+1)^{th}$ standard ordered basis element in $\mathbb{C}^d$. Replacing the Kronecker product $\otimes$ with the outer product $\circ$ yields a cubical tensor of order $n$ with side length $d$: 
\[\sum\limits_{i_1,\dots,i_n=0}^{d-1}\psi_{i_1\dots i_n}\ket{i_1}\circ\dots\circ \ket{i_n} =: \widehat{\psi},\]
which has Frobenius norm $1$. In particular, there is a bijection between the space of $n$-qudit states and the space of cubical hypermatrices of order $n$ with side length $d$ and Frobenius norm $1$ given by the map 
\[\ket{\psi} \mapsto \widehat{\psi}.\]

The concurrence is a popular entanglement measure on $2n$-qubit states that measures the $2n$-way entanglement involving all qubits \cite{coffman2000distributed,wong2001potential}. In \cite{dobes2024qubits}, Dobes \& Jing prove that on pure states of $2n$-qubits, the concurrence $C$ satisfies the equation
\[C(\ket{\psi}) = 2|\mathrm{hdet}(\widehat{\psi})|,\]
and in \cite{dobes2025cayley} they show that more generally on pure states of $2n$-qudits, $|\mathrm{hdet}|^{2/d}$ satisfies the entanglement measure axioms. Thus, Cayley's first hyperdeterminant may be considered a physically meaningful generalization of the concurrence. 

A special class of $n$-qudit states of particular interest to physicists are those found in bosonic quantum systems (see \cite{marconi2025symmetric} for a nice overview). Mathematically, the state space of $n$ indistinguishable bosons is the Hilbert space $\mathrm{Sym}^n(\mathbb{C}^d)$, and so states $\ket{\psi}$ in bosonic quantum systems satisfy the symmetry condition:
\[\ket{\psi} = P_{\pi}\cdot \ket{\psi}\qquad \forall \pi\in S_n,\]
where 
\[P_{\pi}\cdot \ket{\psi} := \sum\limits_{i_1,\dots,i_n=0}^{d-1}\psi_{i_{\pi^{-1}(1)},\dots,i_{\pi^{-1}(n)}}\ket{i_{\pi^{-1}(1)}\dots i_{\pi^{-1}(n)}}.\]
The bijection $\ket{\psi}\mapsto \widehat{\psi}$ implies that $\ket{\psi}$ satisfies the symmetry condition if and only if $\widehat{\psi}$ is symmetric. Therefore, to calculate the $2n$-way entanglement of a symmetric $2n$-qudit state reduces to calculating the hyperdeterminant of its corresponding tensor. \cref{Algorithm 1} therefore provides an efficient way for one to calculate the $2n$-way entanglement of $2n$-qudit states of bosonic quantum systems. 

\section{Conclusion/Future Work}
In summary, in this article we have derived a new formula for computing Cayley's first hyperdeterminant $\mathrm{hdet}$ in terms of a multilinear matrix product involving the Levi-Civita symbol and the canonical vectorization operator. For symmetric tensors of order $N$ with side length $d$, we may consider their $1/N$-hypermatrix vectorization $\mathrm{hvec}_{1/N}$, which contains the same amount of information as their vectorization but requires only $\binom{d+N-1}{N}$ entries rather than $d^N$. This therefore substantially reduces the time complexity for computing $\mathrm{hdet}$ in terms of $N$ assuming $d$ fixed. An important application for this is the efficient calculation of the $2n$-way entanglement for $2n$-qudit states of bosonic quantum systems. 

Besides exploiting the sparsity of $\mathcal{E}_d^{(N)}$ to reduce memory costs and the burden of precomputation, another (more interesting) potential direction for future research could be to consider partially symmetric tensors and derive subgroup analogues of generalized duplication matrices. That is, for subgroups $G\leq S_N$ and \textbf{$\boldsymbol{G}$-symmetric} tensors satisfying the property: 
\[a_{i_{\sigma(1)}\dots i_{\sigma(N)}} = a_{i_1\dots i_N}\qquad \forall \sigma\in G,\]
define $G$-analogues of $\mathrm{hvec}_{1/N}$ and $D_d^{(N)}$. With such objects, one could then factor $\mathrm{hdet}$ on partially symmetric tensors in a similar manner to equation \eqref{Hyperdeterminant Computation - Symmetric Case}, and then for fixed $d$ and precomputed $G$-analogue of $\mathcal{E}_d^{(N)}$, consider the run time cost as $N\rightarrow \infty$. For partially symmetric tensors, the run time cost will be atleast as costly as in the fully symmetric case, but is it sill polynomial for certain proper subgroups (indeed it trivially is for the alternating subgroup $A_N$, but this does not hold in general for all proper subgroups), and which subgroups yield quasi-polynomial or sub-exponential run-time costs? This seems like an interesting problem at the intersection of group theory and complexity theory. Additionally, $G$-analogues of generalized duplication matrices likely have a rich algebraic structure worth further investigating. \\
\quad \\
\noindent{\bf Data availability statement}
Any data that support the findings of this study are included within the article. \\
\\
\noindent{\bf Conflict of interest statement}
The author has no conflicts of interest.

\bibliographystyle{ieeetr}
\bibliography{References}

@book{cayley1844theory,
  title={On the theory of determinants},
  author={Cayley, Arthur},
  year={1844},
  publisher={Pitt Press}
}

@article{zappa1997cayley,
  title={The Cayley determinant of the determinant tensor and the Alon--Tarsi conjecture},
  author={Zappa, Paolo},
  journal={Advances in Applied Mathematics},
  volume={19},
  number={1},
  pages={31--44},
  year={1997},
  publisher={Elsevier}
}

@article{luque2003hankel,
  title={Hankel hyperdeterminants and Selberg integrals},
  author={Luque, Jean-Gabriel and Thibon, Jean-Yves},
  journal={Journal of Physics A: mathematical and general},
  volume={36},
  number={19},
  pages={5267},
  year={2003},
  publisher={IOP Publishing}
}

@article{matsumoto2008hyperdeterminantal,
  title={Hyperdeterminantal expressions for Jack functions of rectangular shapes},
  author={Matsumoto, Sho},
  journal={Journal of Algebra},
  volume={320},
  number={2},
  pages={612--632},
  year={2008},
  publisher={Elsevier}
}

@article{lammers2021generalisation,
  title={A generalisation of the honeycomb dimer model to higher dimensions},
  author={Lammers, Piet},
  journal={The Annals of Probability},
  volume={49},
  number={2},
  pages={1033--1066},
  year={2021},
  publisher={JSTOR}
}

@article{amanov2023tensor,
  title={Tensor slice rank and Cayley's first hyperdeterminant},
  author={Amanov, Alimzhan and Yeliussizov, Damir},
  journal={Linear Algebra and its Applications},
  volume={656},
  pages={224--246},
  year={2023},
  publisher={Elsevier}
}

@article{dobes2024qubits,
  title={Qubits as hypermatrices and entanglement},
  author={Dobes, Isaac and Jing, Naihuan},
  journal={Physica Scripta},
  volume={99},
  number={5},
  pages={055110},
  year={2024},
  publisher={IOP Publishing}
}

@article{dobes2025cayley,
  title={Cayley's First Hyperdeterminant is an Entanglement Measure},
  author={Dobes, Isaac and Jing, Naihuan},
  journal={arXiv preprint arXiv:2504.15511},
  year={2025}
}

@article{hillar2013most,
  title={Most tensor problems are NP-hard},
  author={Hillar, Christopher J and Lim, Lek-Heng},
  journal={Journal of the ACM (JACM)},
  volume={60},
  number={6},
  pages={1--39},
  year={2013},
  publisher={ACM New York, NY, USA}
}

@article{amanov2024enhanced,
  title={Enhanced algorithm for computing cayley’s first hyperdeterminant},
  author={Amanov, Alimzhan},
  journal={Herald of the Kazakh-British technical university},
  volume={21},
  number={3},
  pages={58--65},
  year={2024}
}

@article{lim2013tensors,
  title={Tensors and hypermatrices},
  author={Lim, Lek-Heng},
  journal={Handbook of linear algebra},
  volume={2},
  year={2013},
  publisher={CRC Press Boca Raton, FL}
}

@article{pickard2024kronecker,
  title={Kronecker product of tensors and hypergraphs: structure and dynamics},
  author={Pickard, Joshua and Chen, Can and Stansbury, Cooper and Surana, Amit and Bloch, Anthony M and Rajapakse, Indika},
  journal={SIAM Journal on Matrix Analysis and Applications},
  volume={45},
  number={3},
  pages={1621--1642},
  year={2024},
  publisher={SIAM}
}

@book{gilmore2008lie,
  title={Lie groups, physics, and geometry: an introduction for physicists, engineers and chemists},
  author={Gilmore, Robert},
  year={2008},
  publisher={Cambridge University Press}
}

@article{magnus1980elimination,
  title={The elimination matrix: some lemmas and applications},
  author={Magnus, Jan R and Neudecker, Heinz},
  journal={SIAM Journal on Algebraic Discrete Methods},
  volume={1},
  number={4},
  pages={422--449},
  year={1980},
  publisher={SIAM}
}

@article{mermin1984stirling,
  title={Stirling’s formula},
  author={Mermin, D},
  journal={American J. Physics},
  volume={52},
  pages={362--365},
  year={1984}
}

@article{dur2000three,
  title={Three qubits can be entangled in two inequivalent ways},
  author={D{\"u}r, Wolfgang and Vidal, Guifre and Cirac, J Ignacio},
  journal={Physical Review A},
  volume={62},
  number={6},
  pages={062314},
  year={2000},
  publisher={APS}
}

@article{bennett1996mixed,
  title={Mixed-state entanglement and quantum error correction},
  author={Bennett, Charles H and DiVincenzo, David P and Smolin, John A and Wootters, William K},
  journal={Physical Review A},
  volume={54},
  number={5},
  pages={3824},
  year={1996},
  publisher={APS}
}

@article{vedral1997quantifying,
  title={Quantifying entanglement},
  author={Vedral, Vlatko and Plenio, Martin B and Rippin, Michael A and Knight, Peter L},
  journal={Physical Review Letters},
  volume={78},
  number={12},
  pages={2275},
  year={1997},
  publisher={APS}
}

@article{vidal2000entanglement,
  title={Entanglement monotones},
  author={Vidal, Guifr{\'e}},
  journal={Journal of Modern Optics},
  volume={47},
  number={2-3},
  pages={355--376},
  year={2000},
  publisher={Taylor \& Francis}
}

@article{eltschka2014quantifying,
  title={Quantifying entanglement resources},
  author={Eltschka, Christopher and Siewert, Jens},
  journal={Journal of Physics A: Mathematical and Theoretical},
  volume={47},
  number={42},
  pages={424005},
  year={2014},
  publisher={IOP Publishing}
}

@book{bengtsson2017geometry,
  title={Geometry of quantum states: an introduction to quantum entanglement},
  author={Bengtsson, Ingemar and {\.Z}yczkowski, Karol},
  year={2017},
  publisher={Cambridge university press}
}

@article{coffman2000distributed,
  title={Distributed entanglement},
  author={Coffman, Valerie and Kundu, Joydip and Wootters, William K},
  journal={Physical Review A},
  volume={61},
  number={5},
  pages={052306},
  year={2000},
  publisher={APS}
}

@article{wong2001potential,
  title={Potential multiparticle entanglement measure},
  author={Wong, Alexander and Christensen, Nelson},
  journal={Physical Review A},
  volume={63},
  number={4},
  pages={044301},
  year={2001},
  publisher={APS}
}

@article{marconi2025symmetric,
  title={Symmetric quantum states: a review of recent progress},
  author={Marconi, Carlo and M{\"u}ller-Rigat, Guillem and Romero-Pallej{\`a}, Jordi and Tura Brugu{\'e}s, Jordi and Sanpera, Anna},
  journal={Reports on Progress in Physics},
  year={2025}
}

\end{document}